\documentclass[12pt]{iopart}

\usepackage{setstack,amsthm,amsfonts,amssymb,times,bbm,bm,graphicx,dsfont,bbold}
\usepackage{etoolbox}
\usepackage{xcolor}
\usepackage{hyperref}
\usepackage{cite}
\usepackage{array}
\newcolumntype{L}{>{$}l<{$}} 
\usepackage{braket}
\usepackage{latexsym}
\usepackage[T1]{fontenc}
\usepackage[utf8x]{inputenc}

\newtheorem{theorem}{Theorem}

\newtheorem{lemma}[theorem]{Lemma}
\newtheorem{corollary}[theorem]{Corollary}
\newtheorem{conjecture}[theorem]{Conjecture}
\newtheorem{definition}[theorem]{Definition}

\newtheorem{example}[theorem]{Example}

\newcommand{\eqref}[1]{(\ref{#1})}
\newcommand{\bd}[1]{\boldsymbol{ #1 }}
\newcommand{\BDv}[3]{\mbox{$\scriptsize{\left(\!\!\begin{array}{c}#1\\#2 \\#3\end{array}\!\!\right)}$}}
\renewcommand\bra[1]{{\langle{#1}|}}
\renewcommand\ket[1]{{|{#1}\rangle}}

\usepackage{subfig}

\def\RR{\mathbbm{R}}
\def\CC{\mathbbm{C}}

\def\H{\mathcal{H}}
\def\Ho{\mathcal{H}_1}
\def\HN{\wedge^N[\Ho]}

\def\kpsi {\ket{\Psi}}

\begin{document}

\title[]{Implications of pinned occupation numbers for natural orbital expansions. I: Generalizing the concept of active spaces}

\author{Christian Schilling}
\ead{c.schilling@lmu.de}
\address{Department of Physics, Arnold Sommerfeld Center for Theoretical Physics, Ludwig-Maximilians-Universit\"at M\"unchen, Theresienstrasse 37, 80333 M\"unchen, Germany}
\address{Clarendon Laboratory, University of Oxford, Parks Road, Oxford OX1 3PU, United Kingdom}

\author{Carlos L. Benavides-Riveros}
\address{Institut f\"ur Physik, Martin-Luther-Universit\"at Halle-Wittenberg, 06120 Halle (Saale), Germany}


\author{Alexandre Lopes}
\address{Carl Zeiss SMT GmbH, Rudolf-Eber-Straße 2, 73447 Oberkochen, Germany}

\author{Tomasz Maci\k{a}\.zek}
\address{School of Mathematics, University of Bristol, Bristol BS8 1TW, UK}
\address{Center for Theoretical Physics, Polish Academy of Sciences, Al.~Lotnik\'ow 32/46, 02-668 Warszawa, Poland}

\author{Adam Sawicki}
\address{Center for Theoretical Physics, Polish Academy of Sciences, Al.~Lotnik\'ow 32/46, 02-668 Warszawa, Poland}

\date{\today}

\begin{abstract}
The concept of active spaces simplifies the description of interacting quantum many-body systems by restricting to a neighbourhood of active orbitals around the Fermi level. The respective wavefunction ansatzes which involve all possible electron configurations of active orbitals can be characterized by the saturation of a certain number of Pauli constraints $0 \leq n_i \leq 1$, identifying the occupied core orbitals ($n_i=1$) and the inactive virtual orbitals ($n_j=0$). In Part I, we generalize this crucial concept of active spaces by referring to the generalized Pauli constraints. To be more specific, we explain and illustrate that the saturation of any such constraint on fermionic occupation numbers characterizes a distinctive set of active electron configurations. A converse form of this selection rule establishes the basis for corresponding  multiconfigurational wavefunction ansatzes. In Part II, we provide rigorous derivations of those findings. Moroever,
we extend our results to non-fermionic multipartite quantum systems, revealing that extremal single-body information has always strong implications for the multipartite quantum state. In that sense, our work also confirms that pinned quantum systems define new physical entities and the presence of pinnings reflect the existence of (possibly hidden) ground state symmetries.
\end{abstract}


\maketitle

\section{Introduction}\label{sec:intro}
At first sight, an accurate description of fermionic quantum many-body systems seems to be highly challenging, if not impossible: The interaction between the particles can lead to strong correlations which in principle may distribute over an exponentially large Hilbert space. Yet, realistic physical systems exhibit some additional structure. To name possibly the most important one, the particles interact only by \emph{two}-body forces and the respective ground state problem can therefore be addressed in the reduced \emph{two}-particle picture \cite{ColemanRev,Coleman}. Since most subcommunities restrict to systems all characterized by the same pair interaction (for instance Coulomb interaction in quantum chemistry, contact interaction in quantum optics and Hubbard interaction in solid state physics) the ground state problem should \emph{de facto} involve only the one-particle reduced density matrix. Indeed, for Hamiltonians of the form $H_{\kappa}(h) = h + \kappa V$, where $h$ represents the one-particle terms and $V$ the \emph{fixed} pair interaction with coupling strength $\kappa$, the conjugate variable to $H_{\kappa}(h)$ and $h$, respectively, is the one-particle reduced density operator $\rho_1$.  The corresponding exact one-particle theory is known as Reduced Density Matrix Functional Theory (RDMFT) and is based on the existence of an exact energy functional $\mathcal{E}_{\kappa}(\rho_1) \equiv \mbox{Tr}[h \rho_1]+ \kappa \mathcal{F}(\rho_1)$ \cite{Gilbert} (see also \cite{PG16}). Here, the interaction functional $\mathcal{F}(\rho_1)$ is universal in the sense that it depends only on the fixed interaction $V$ but not on the coupling $\kappa$ or the one-particle terms $h$.

There is also another less profound motivation for the description of quantum many-body systems in the one-particle picture, as governed by the one particle reduced density operator. Whenever, the coupling $\kappa$ between the identical fermions vanishes the respective Hamiltonian $H_{\kappa=0}(h)$ contains only one-particle terms and the ground state problem can be entirely discussed and solved in the much simpler one-particle picture: In a first step one needs to diagonalize the one-particle Hamiltonian $h$ on the one-particle Hilbert space $\Ho$,
$h|_{\Ho}= \sum_{j \geq 1} \varepsilon_j \ket{\chi_j} \!\bra{\chi_j}$. Then, in a second step, the $N$ energetically lowest one-particle eigenstates $\ket{\chi_1},\ldots,\ket{\chi_N}$ are occupied successively from below just obeying Pauli's exclusion principle. The respective $N$-fermion ground state follows immediately as
$\ket{\chi_1,\chi_2,\ldots,\chi_N} \equiv f_{\chi_1}^\dagger \ldots f_{\chi_N}^\dagger \ket{vac}$, emphasizing clearly why \emph{configuration states} $\ket{\chi_1,\chi_2,\ldots,\chi_N}$ are considered as being uncorrelated.
But how about interacting systems? By turning on the coupling $\kappa$, the occupation numbers $n_{\chi_i}$ of the individual one-particle states $\ket{\chi_i}$ begin to deviate from the extremal values one and zero, respectively. In other words, the corresponding $N$-fermion ground state $\ket{\Psi(\kappa)} \in \HN$ is not uncorrelated anymore and instead follows in general as a superposition involving various $N$-fermion configurations $(i_1 \ldots,i_N)$
\begin{equation}\label{PsiCI}
\ket{\Psi(\kappa)}  = \sum_{i_1 < \ldots <i_N } c_{i_1,\ldots,i_N}(\kappa)\,\ket{\chi_{i_1},\ldots,\chi_{i_N}}\,.
\end{equation}
This superposition could involve in principle all ${d\choose N}$ configurations of $N$ fermions distributed over $d\equiv\mbox{dim}(\Ho)$ many orbitals $\ket{\chi_j}$. Yet, for realistic systems of confined fermions (e.g., electrons in atoms) the one-particle Hamiltonian $h$ often dominates the interaction Hamiltonian $\kappa V$ and energetically lower or higher lying orbitals $\ket{\chi_j}$ far away from the Fermi level are either almost occupied ($n_{\chi_j}\approx 1$) or almost unoccupied ($n_{\chi_j}\approx 0$). This emphasizes the  significance of the concept of active spaces. To be more specific, it allows one to exploit significantly simplified ansatzes for $\ket{\Psi}$ involving only configurations $i_1 < \ldots <i_N $ with a certain number of fully frozen (core) orbitals and some inactive virtual orbitals.  In quantum chemistry such ground state ansatzes are referred to as \emph{Complete Active Space Self-Consistent Field} (CASSCF) ansatz (see, e.g., Refs.~\cite{CASSCF1,CASSCF2,CASSCF3,CASSCF4}).

The general aim of our paper is to illustrate and prove in a mathematically rigorous way that also the saturation of the \emph{generalized} Pauli constraints (\emph{pinning}) \cite{BD,Kly4,Kly3,Altun} gives rise to specific, generalized active spaces. In that sense, our work shall provide the foundation for possible future applications of the new concept of generalized Pauli constraints within quantum chemistry and physics, particularly in the form of more systematic \emph{Multiconfiguration Self-Consistent Field} (MCSCF) ansatzes. The paper therefore consists of two complementary parts. Part I explains and comprehensively illustrates various results in the context of fermionic quantum systems and avoids any technicalities. Quite in contrast, Part II provides rigorous derivations of our results and extends them to non-fermionic systems.

The present Part I is structured as follows. After fixing the notation and introducing the basic concepts in Section \ref{sec:notation}, we illustrate in Section \ref{sec:PC} the connection of pinning of Pauli constraints and structural simplifications of the $N$-fermion quantum state. This link between the one-particle and $N$-particle picture provides in particular a solid foundation for the concept of (complete) active spaces. In Section \ref{sec:GPC}, we explain and illustrate how this concept of active spaces could be generalized. To be more specific, we present and illustrate our main results stating that the saturation of the generalized Pauli constraints implies a selection rule identifying the $N$-fermion configurations contributing in a respective natural orbital expansion.  A converse form of this selection rule establishes the basis for corresponding  multiconfigurational wavefunction ansatzes.

\section{Notation and concepts}\label{sec:notation}
In the following, we fix the notation and introduce some basic concepts. To keep our work self-contained we in particular recall some concepts which were already introduced and discussed in \cite{CSQuasipinning,CSQ}.
In our work, we always consider a finite $d$-dimensional one-particle Hilbert space $\Ho$. In the context of numerical approaches in physics and quantum chemistry, such $\Ho$ typically arises from the truncation of the full infinite-dimensional one-particle Hilbert space of square integrable wave functions $ L^2(\mathcal{C})\otimes \CC^{2s+1}$ by choosing a finite basis set of $d$ spin-orbitals. A prime example would be electrons in an atom, i.e., spin $s=\frac{1}{2}$ with the underlying configuration space $\mathcal{C}$ given by $\mathcal{C}\equiv \RR^3$ and a basis set of $d$ atomic spin-orbitals.

\subsection{Natural orbitals and natural occupation numbers}\label{subsec:1RDM}
The crucial object of our work is the one-particle reduced density operator $\rho_1$ of an $N$-fermion quantum state $\ket{\Psi}\in \HN$. There are two equivalent routes that one could follow for introducing $\rho_1$. By exploiting first quantization, one naturally embeds the $N$-fermion Hilbert space $\HN \leq \Ho^{\otimes^N}$ into the Hilbert space $ \Ho^{\otimes^N}$ of $N$ distinguishable particles. Tracing out $N-1$ of those tensor product factors $\Ho$ yields
\begin{equation}\label{1RDO}
\rho_1 \equiv N\,\mbox{Tr}_{N-1}[|\Psi\rangle\!\langle\Psi|]\,.
\end{equation}
The partial trace in Eq.~(\ref{1RDO}) is indeed well-defined since the choice of the $N-1$ factors to be traced out does not matter due to the well-defined exchange-symmetry of $\ket{\Psi}$. An alternative but equivalent approach to define $\rho_1$ is based on second quantization. After fixing some orthonormal reference basis $\{\ket{\chi_j}\}_{j=1}^d$ for the one-particle Hilbert space $\Ho$ and introducing the respective creation and annihilation operators, $\rho_1$ follows from its matrix representation
\begin{equation}\label{1RDM}
  \bra{\chi_i} \rho_1 \ket{\chi_j} \equiv \bra{\Psi}f_{\chi_j}^\dagger f_{\chi_i} \ket{\Psi}\,.
\end{equation}
Diagonalizing the Hermitian one-particle reduced density operator $\rho_1$,
\begin{equation}\label{1RDOdiag}
\rho_1 = \sum_{j=1}^d\,n_j\,|j\rangle\! \langle j|\,,
\end{equation}
gives rise to the \emph{natural occupation numbers} (NONs) $n_j$ and the \emph{natural orbitals} $|j\rangle$, the corresponding eigenstates \cite{Loewdin55,DavidsonRev}. This terminology also motivates the normalization $\mbox{Tr}_1[\rho_1]=n_1+\ldots+n_d =N$ which allows us to interpret the eigenvalues of $\rho_1$ as occupation numbers, the occupancies of the natural orbitals.
Moreover, for the following considerations we order the NONs decreasingly, $n_1\geq n_2\geq\ldots\geq n_d\geq 0$.

The natural orbitals of any $N$-fermion state $\ket{\Psi}$ form an orthonormal basis $\mathcal{B}_1=\{|j\rangle\}_{j=1}^d$ for the one-particle Hilbert space $\Ho$. This basis is unique (up to phases) as long as the NONs are non-degenerate.
Based on the natural orbital basis $\mathcal{B}_1$, we introduce a natural orbital induced operator which will play a crucial role for the compact formulating of our main results:
\begin{definition}[Natural orbital induced operators]\label{def:Lhat}
Given $\ket{\Psi}\in \HN$ and let $\mathcal{B}_1$ be a basis of natural orbitals. For any polynomial $L$ of $d$ variables of degree one, we define
\begin{equation}
\hat{L}_{\mathcal{B}_1} \equiv L(\hat{n}_1,\ldots,\hat{n}_d)\,,
\end{equation}
where the particle number operators $\hat{n}_j \equiv f_j^\dagger f_j$ refer to the natural orbitals $\mathcal{B}_1$.
\end{definition}
\noindent Since we use this concept of an orbital induced operator only with respect to the natural orbitals of a given quantum state $\ket{\Psi}$ we refrained from extending the definition of $\hat{L}_{\mathcal{B}_1}$ to arbitrary  orthonormal bases $\mathcal{B}_1$. We would also like to stress again that here and in the following, the natural orbitals $\ket{j}$ of $\ket{\Psi}$ are only unique as long as the NONs are non-degenerate and their labelling resembles that of the corresponding NONs, i.e.~$n_1\geq n_2 \geq \ldots \geq n_d$.

Of course, we could have easily extended the Definition \ref{def:Lhat} to all analytic functions of $d$ variables. Yet, only for linear forms $L$ the following important identity holds (due to $n_j = \bra{\Psi}\hat{n}_j\ket{\Psi}$)
\begin{equation}
\bra{\Psi} \hat{L}_{\mathcal{B}_1}\ket{\Psi}\equiv \bra{\Psi} L(\hat{n}_1,\ldots,\hat{n}_d)\ket{\Psi}  = L(n_1,\ldots,n_d)\,.
\end{equation}

\subsection{Natural orbital expansion}\label{subsec:Psi}
In general, any orthonormal basis $\mathcal{B}_1\equiv \{\ket{i}\}_{i=1}^d$ for $\Ho$ induces an orthonormal basis $\mathcal{B}_N$ for $\HN$, given by the family of configuration states
\begin{equation}
\ket{\bd{i}} \equiv \ket{i_1,i_2,\ldots, i_N} \equiv f_{i_1}^\dagger f_{i_2}^\dagger  \ldots f_{i_N}^\dagger \ket{vac}\,,
\end{equation}
where $1\leq i_1<i_2< \ldots <i_N\leq d$, $\bd{i}\equiv (i_1,\ldots,i_N)$ and $\ket{vac}$ denotes the vacuum state of the Fock space constructed over $\H_1$.  For ease of notation we suppress here and in the following the explicit dependence of the configuration states $\ket{\bd{i}} \in \mathcal{B}_N$ on $|\Psi\rangle$ and the choice $\mathcal{B}_1$ of natural orbitals in case the NONs are degenerate.
Since $\mathcal{B}_N$ is a basis for $\HN$ we can expand every quantum state in $\HN$ uniquely with respect to $\mathcal{B}_N$, in particular also $|\Psi\rangle$ itself (whose natural orbitals gave rise to $\mathcal{B}_1$ and thus $\mathcal{B}_N$) \cite{Loewdin56,DavidsonRev}
\begin{equation}\label{PsiNO}
\ket{\Psi} = \sum_{\bd{i}}\, c_{\bd{i}}\,\ket{\bd{i}}\,.
\end{equation}
Notice that this expansion based on natural orbitals imposes quite strong restrictions on the expansion coefficients $c_{\bd{i}}$.
These self-consistency conditions namely reflect the fact that the corresponding one-particle reduced density operator $\rho_1$ (\ref{1RDO}) is diagonal with respect to its own natural orbitals $\ket{j}$. In addition, the occupancy of $\ket{j}$
is given by $n_j$, the $j$-th largest NON,
\begin{equation}\label{NONc}
n_j = \sum_{\bd{i}\ni j} |c_{\bd{i}}|^2\,.
\end{equation}
Note, that in a natural expansion (\ref{PsiNO}) some of the coefficients may be zero. We will often distinguish the set of configuration states which do not contribute to the expansion of $\kpsi$ and call it the {\it natural support}, $\mathrm{Supp}_{\mathcal{B}_1}(\kpsi)$, of $\kpsi$
\begin{equation}\label{def:support}
\mathrm{Supp}_{\mathcal{B}_1}(\kpsi):=\{\bd{i}:\ \ket{\bd i}\in\mathcal{B}_N{\rm\ and}\ \braket{\Psi|\bd i}\neq0\}.
\end{equation}
Clearly, in case of degenerate NONs the support of $\kpsi$ may depend on the specific choice $\mathcal{B}_1$ of natural orbitals.

\subsection{Geometric picture of occupation numbers}\label{subsec:geom}
Equation~\eqref{NONc} allows us to interpret the self-consistent expansion \eqref{PsiNO} geometrically. By denoting for each configuration state $\ket{\bd{i}}$ the respective vector of \emph{unordered} occupation numbers by $\bd{n}_{\bd{i}}$,
\begin{equation}\label{NONconf}
  \bd{n}_{\bd{i}}\equiv \mbox{spec}\big(N \mbox{Tr}_{N-1}[\ket{\bd{i}}\!\bra{\bd{i}}]\big)\,,\quad \mbox{i.e.},\,\left(\bd{n}_{\bd{i}}\right)_j=
 \Bigg\{ \begin{array}{@{\kern2.5pt}lL}
    \hfill 1 & if $j \in \bd{i}$\\
          0 & if $j \not \in \bd{i}$
\end{array}\,,
\end{equation}
Eq.~\eqref{NONc} implies
\begin{equation}\label{NONgeom}
\bd{n} = \sum_{\bd{i}}\, |c_{\bd{i}}|^2\, \bd{n}_{\bd{i}}\,.
\end{equation}
This means that the vector $\bd{n}$ of NONs follows as the ``center of mass'' for masses $|c_{\bd{i}}|^2$ located at positions $\bd{n}_{\bd{i}}$ in $\RR^d$. Since each $\bd{n}_{\bd{i}}$ contains $N$ ones and $d-N$ zeros, the vectors $\bd{n}_{\bd{i}}$ are vertices of
the Pauli hypercube $[0,1]^d$, namely exactly those with normalization $\|\bd{n}_{\bd{i}}\|_1=N$. All the other vertices of the Pauli hypercube $[0,1]^d$ would correspond to configuration states of particle numbers different than $N$ and therefore will not play any role in the present work which restricts to fixed particle number $N$. This geometric picture is illustrated in Figure \ref{fig:polytope} in Section \ref{subsec:pin} for the Borland-Dennis setting, i.e., for the case of three fermions and a six-dimensional one-particle Hilbert space.

Lastly, we point out a geometric aspect concerning the action of operators $\hat L_{\mathcal{B}_1}$ from Definition \ref{def:Lhat} on configuration states $\ket{\bd i}$. Namely, for a given $L=\sum_{j=1}^d l_j n_j$ it is straightforward to check that $\hat L_{\mathcal{B}_1}$ is diagonal in the NO-basis and that its diagonal entries follow by the  geometric formula
\begin{equation}\label{eq:Lhat-action}
\hat L_{\mathcal{B}_1}\ket{\bd i}=\bd{l}\cdotp\bd{n_{\bd i}}\,\ket{\bd i}.
\end{equation}
Here, we use the standard notation for the dot-product of vectors, i.e.  $\bd{l}\cdotp\bd{n_{\bd i}}:=\sum_{j=1}^d l_j (\bd{n_{\bd i}})_j$.

\section{Pauli constraints and concept of active spaces}\label{sec:PC}
The properties and the behavior of fermionic quantum systems strongly rely on Pauli's exclusion principle \cite{Pauli1925}.
This principle defines a constraint on the one-particle picture as governed by the one-particle reduced density operator $\rho_1$. For any $N$-fermion state $\ket{\Psi}$ the occupancies of one-particle states $\ket{\varphi}$  are restricted, $0\leq \bra{\Psi} \hat{n}_{\varphi}\ket{\Psi}\leq 1$. Indeed, since $\bra{\Psi} \hat{n}_{\varphi}\ket{\Psi} = \bra{\varphi} \rho_1\ket{\varphi}$ this constrains $\rho_1$ according to
\begin{equation}\label{PCrho}
\mathbb{0}\leq \rho_1\leq \mathbb{1}  \,.
\end{equation}
Equivalent to this operator relation, the NONs $n_i$ (eigenvalues of $\rho_1$) are restricted,
\begin{equation}\label{PC}
0 \leq n_i \leq 1\,.
\end{equation}
These Pauli constraints play an important role for various physical phenomena with remarkable consequences for both, the microscopic and the macroscopic world. On a microscopic length scale, they are the basis of the Aufbau principle for atoms and nuclei. For macroscopic systems the Pauli exclusion principle is responsible for the very stability of matter \cite{Dyson67,LiebStab}. This universal relevance of Pauli's exclusion principle is quite obvious for weakly interacting systems: All Pauli constraints are (approximately) saturated, i.e.,~one observes for each NON either $n_i \approx 1$ or $n_i \approx 0$. Such (approximate) \emph{pinning} of all Pauli constraints is the typical behavior within mean field theories such as the Landau-Fermi theory or the Hartree-Fock theory. Even for strongly correlated systems one often observes this quasipinning by Pauli constraints since at least the largest occupation numbers are very close to one and the smallest ones are very close to zero. For instance, the 1s shell in atoms (under realistic conditions) is typically fully occupied and the normalization $\sum_{i=1}^d n_i=N$ requires the smallest NONs to be arbitrarily small for large or even infinite basis set size $d\gg N$.

In the following, we would like to formalize the concept of active spaces by relating their structure in the $N$-particle picture to the possible saturation of multiple Pauli constraints concerning the one-particle picture. For this, we express the family of Pauli constraints \eqref{PC} in a  more compact form. For any pair $(r,s)$ of integers $0 \leq r\leq N$, $0 \leq s\leq d-N$ we define the constraints (see also \cite{CSQ})
\begin{equation}\label{PCrs}
S^{(r,s)}(\bd{n})\equiv\sum_{i=1}^r(1-n_i)+\sum_{j=d+1-s}^d n_j \,\geq\, 0\,
\end{equation}
on the non-increasingly ordered NONs $\bd{n}$. The family of those constraints is equivalent to the Pauli constraints in their original form \eqref{PC}. From the geometric point of view (recall Section \ref{subsec:geom}), all vectors $\bd{n}$ of non-increasingly ordered NONs obeying the Pauli exclusion principle form a specific polytope in $\RR^d$, the Pauli simplex $\Sigma$,
\begin{equation}\label{PCsimplex}
\Sigma \equiv \left\{\bd{n}\in \RR^d\,|\, 1\geq n_1\geq n_2 \geq \ldots \geq n_d \geq 0\, \wedge \, \|\bd{n}\|_1=N \right\}\,.
\end{equation}

\begin{theorem}[Active space]\label{thm:CAS}
Let $\ket{\Psi} \in \HN$ with $\mbox{dim}(\Ho)=d$, recall Definition \ref{def:Lhat} and let $\mathcal{B}_1$ be a basis of natural orbital of $\ket{\Psi}$. For all integers $r\leq N$, $s\leq d-N$ one then has
\begin{equation}\label{CAS0}
S^{(r,s)}(\bd{n}) = 0 \quad \Leftrightarrow \quad \hat{S}_{\mathcal{B}_1}^{(r,s)}\ket{\Psi} = 0\,.
\end{equation}
This implies a selection rule on the expansion coefficients in the sense that only those configuration states  $\ket{\bd{i}}\equiv\ket{i_1,\ldots,i_N}$ may  contribute to the self-consistent expansion of $\ket{\Psi}$ (recall \eqref{PsiNO}) which include all natural orbitals $\ket{1},\ldots,\ket{r}$ and exclude $\ket{d-N+1},\ldots,\ket{d}$. To be more precise, this means
\begin{equation}\label{SRrs}
S^{(r,s)}(\bd{n}_{\bd{i}})\neq 0\quad \Rightarrow \quad c_{\bd{i}} = 0\,,
\end{equation}
where $\bd{n}_{\bd{i}}$ is the unordered spectrum of the configuration state $\ket{\bd{i}}$ as introduced in Eq.~\eqref{NONconf}.
\end{theorem}
\begin{proof}
Since $S^{(r,s)}(\bd{n})= \bra{\Psi}\hat{S}_{\mathcal{B}_1}^{(r,s)}\ket{\Psi}$ the direction ``$\Leftarrow$'' in Eq.~\eqref{CAS0} follows immediately. To prove ``$\Rightarrow$'', we observe that the configuration states $\ket{\bd{i}}$ are the eigenstates of the operator $\hat{S}_{\mathcal{B}_1}^{(r,s)}$ with respective integer eigenvalues $S^{(r,s)}(\bd{n}_{\bd{i}})$. Since the smallest eigenvalue is zero, $S^{(r,s)}(\bd{n}) =0$ implies that the whole weight of $\ket{\Psi}$ needs to lie in the zero eigenspace. The Selection Rule \eqref{SRrs} follows then immediately by plugging in the expansion \eqref{PsiNO} into \eqref{CAS0} and using again the fact that $\hat{S}_{\mathcal{B}_1}^{(r,s)}$ is diagonal with respect to the configuration states $\ket{\bd{i}}$.
\end{proof}

The proof of Theorem \ref{thm:CAS} and the derivation of the consequences of pinning by the Pauli constraints, respectively, was rather elementary. This is due to the fact that the natural orbital induced operator $\hat{S}_{\mathcal{B}_1}^{(r,s)}$ has no negative eigenvalues, i.e.~it is positive semi-definite.
Therefore, whenever $S^{(r,s)}(\bd{n})=\bra{\Psi}\hat{S}_{\mathcal{B}_1}^{(r,s)}\ket{\Psi}=0$, $\ket{\Psi}$ cannot have any weight in eigenspaces with positive eigenvalues since their contributions to $\bra{\Psi}\hat{S}_{\mathcal{B}_1}^{(r,s)}\ket{\Psi}$ could not be cancelled out by contributions from eigenspaces with negative eigenvalues. This will be different when we discuss in the following the consequences of pinning of generalized Pauli constraints, $D(\bd{n})= 0$, since their respective natural orbital induced operators $\hat{D}_{\mathcal{B}_1}$ have both negative and positive eigenvalues.

\section{Generalized Pauli constraints and generalized active spaces}\label{sec:GPC}

\subsection{Generalized Pauli constraints}\label{subsec:GPC}
Despite the remarkable significance of Pauli's exclusion principle \eqref{PC}, \eqref{PCrs} on all physical length scales, it has conclusively been shown only recently \cite{Kly2,Kly3,Altun} that the fermionic exchange symmetry implies even greater restrictions on the one-particle picture. To be more specific, as illustrated in Fig.~\ref{fig:fqmp}, the set of pure $N$-representable vectors $\bd{n}$ of (non-increasingly ordered) NONs form a polytope, a proper subset of the Pauli simplex $\Sigma$ \eqref{PCsimplex}.
\begin{figure}[h]
\centering
\includegraphics[scale=0.8]{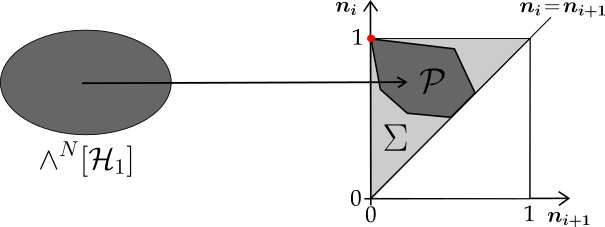}
\caption{Schematic illustration of the map assigning to each $N$-fermion quantum state $\ket{\Psi}\in\HN$ its vector $\bd{n}\in \RR^d$ of non-increasingly ordered NONs. The set of all attainable $\bd{n}$, as described by the generalized Pauli constraints \eqref{GPC}, forms a polytope $\mathcal{P}$, a proper subset of the Pauli simplex $\Sigma$ \eqref{PCsimplex} (shown in light-gray).  The ``Hartree-Fock point'', corresponding to ground states of non-interacting fermions, is shown as a red dot.}
\label{fig:fqmp}
\end{figure}
For each setting of $N$ fermions and a $d$-dimensional one-particle Hilbert space $\Ho$, this polytope $\mathcal{P}$ is described by a finite family of linear inequalities, the \emph{generalized Pauli constraints} (GPC),
\begin{equation}\label{GPC}
  D_i(\bd{n}) \equiv \kappa_i^{(0)} +  \bd{\kappa}_i \cdot \bd{n} \equiv \kappa_i^{(0)} + \sum_{j=1}^d \kappa_i^{(j)} n_j\geq 0\,,
\end{equation}
$i=1,2,\ldots, r_{N,d}<\infty$. For each GPC $D_i \geq 0$, the respective coefficients $\kappa_i^{(j)}$ can be chosen as integers. In particular, by referring to the canonical choice of \emph{minimal} integers, the $l^1$-distance of $\bd{n}$ to the hyperplane defined by $D_i\equiv0$ follows as $D_j(\bd{n})$ up to a prefactor (for more details see Ref.~\cite{CS2016a}).

While the GPCs for the smaller settings with $N,d \leq 7$ have already been derived several decades ago by some brute force approach  \cite{Smith,BD,Rus2}, it was Klyachko's breakthrough \cite{Kly2,Kly3} on how to find a systematic procedure which allows one to determine for all settings $(N,d)$, at least in principle, the family of GPCs. Yet, it is still an ongoing challenging to develop more efficient algorithms for determining the GPCs and  in particular to approximate them (see, e.g., Ref.~\cite{TomekDoub}).
Before we briefly discuss the potential physical relevance of the GPCs, we would like to present them for the first non-trivial  setting, $(N,d)=(3,6)$, and comment on their triviality for the smallest few settings.

First, due to the particle hole duality on the fermionic Fock space we can restrict ourselves without loss of generality to $N \leq d/2$. Indeed, one has (see, e.g., \cite{Kly2})
\begin{lemma}[Particle-hole duality]\label{lem:PHduality}
The generalized Pauli constraints of the setting $(d-N,d)$ of $d-N$ fermions and a $d$-dimensional one-particle Hilbert space follows from those of $(N,d)$ by just replacing $n_i \mapsto 1- n_{d-i+1}$ for all $i$.
\end{lemma}
\noindent Second, as summarized by Example \ref{ex:trivial}, the GPCs for all settings with only one or two fermions (and according to the particle-hole duality, Lemma \ref{lem:PHduality}, also those with one or two holes) are trivial \cite{BD2}. The first non-trivial setting is thus the
Borland-Dennis setting, i.e.~$(N,d)=(3,6)$.
\begin{example}[Trivial settings]\label{ex:trivial}
The GPCs for $N=1$ are given by $n_1=1$ and $n_i=0$ for all $i \geq 2$ (i.e., the polytope of mathematically possible $\bd{n}$ contains only one point). For the case of $N=2$ fermions, the GPCs are given by $n_{2k-1}=n_{2k}$ for all $k \leq \lfloor d/2 \rfloor$ and in case $d$ is odd one has additionally $n_d=0$.
\end{example}
\begin{example}[Borland-Dennis setting]\label{ex:BD}
The GPCs for the setting $(N,d)=(3,6)$ read \cite{BD}
\begin{eqnarray}
&&1-(n_1+n_6) = 1-(n_2+n_5) = 1-(n_3+n_4) = 0\,, \label{d=6a} \\
&&D(\bd{n}) \equiv 2-(n_1 + n_2+ n_4) \geq 0 \label{d=6b} \,.
\end{eqnarray}
\end{example}
\noindent We remind the reader that the NONs are always ordered non-increasingly, $n_1 \geq n_2 \geq \ldots \geq n_d \geq 0$.
Notice that the inequality $D(\bd{n})\geq 0$ is more restrictive than Pauli's exclusion principle, which just states implies $2-(n_1 + n_2)\geq 0$. The incidence of GPCs taking the form of equalities (instead of inequalities) as those in (\ref{d=6a}) is rather unique since this
happens only for the Borland-Dennis setting and the settings with at most two fermions or at most two holes.

\subsection{Potential physical relevance of the generalized Pauli constraints}\label{subsec:GPCrel}
In complete analogy to Pauli's exclusion principle, the physical significance of the GPCs is primarily be based on their possible (approximate) saturation in concrete systems.
In an analytical study \cite{CS2013} of the ground state of three harmonically interacting fermions in a one-dimensional harmonic trap it has been shown that the GPCs are not fully saturated. Yet, given this it is quite remarkable that the vector $\bd{n}$ of NONs has just a tiny distance to the polytope boundary given by the eighth power of the coupling strength, $D\propto \kappa^8$. A succeeding comprehensive and conclusive study of harmonic trap systems \cite{Ebler,CSthesis,CS2016a,CS2016b,FTthesis,CSpair} has confirmed that such \emph{quasipinning} represents a genuine physical effect whose origin is the universal conflict between energy minimization and fermionic exchange symmetry in systems of confined fermions \cite{CS2016b}. The presence of such quasipinning (or even pinning if the system's chosen Hilbert space is artificially small) has been verified also in smaller atoms and molecules  \cite{Kly1,BenavLiQuasi,Kly5,Mazz14,MazzOpen,BenavQuasi2,Mazzagain,Alex,CS2015Hubbard,CBthesis,NewMazziotti,MazzOpen2,CSQChem}
). A comment is in order concerning the non-triviality of such (quasi)pinning by the GPCs. Since at least some NONs in most realistic ground states are close to one, the vector $\bd{n}\in \mathcal{P}$ of NONs is typically close to the boundary of the surrounding Pauli simplex $\Sigma$ \eqref{PCsimplex} and consequently (recall $\mathcal{P}\subset \Sigma$ and see Fig.~\ref{fig:fqmp}) it is also close to the boundary of the polytope $\mathcal{P}$. The more crucial question is therefore whether the (quasi)pinning by the GPCs is nontrivial in the sense that it does not already follow from (quasi)pinning by the Pauli constraints, or in other words, whether the GPCs have any significance beyond the Pauli constraints \eqref{PCrs}. This also necessitates a systematic treatment of systems with symmetries, since symmetries are known to favour the occurrence of rather artificial (quasi)pinning \cite{BenavQuasi2,CS2015Hubbard,CBthesis}. A more systematic recent analysis based on the so-called $Q$-parameter \cite{CSQ} has shown that the quasipinning by the GPCs is indeed non-trivial \cite{CSQ,CSpair,CSQChem}.

It has been speculated and suggested that such (quasi)pinning would reduce the complexity of the system's quantum state and would define \emph{``a new physical entity with its own dynamics and kinematics''} \cite{Kly1} (see also \cite{CSQMath12,CSQuasipinning,Stability}).
Based on this expected implication of (quasi)pinning as an effect in the one-particle picture on the structure of the $N$-fermion quantum states, variational ansatzes for ground states have been proposed as part of an ongoing development \cite{CSHF,Stability,MazzSparse1,CBfunct,MazzSparse2}. Moreover, general investigations and deeper insights into the structure of quantum states suggest that taking the GPCs into account may help to turn Reduced Density Matrix Functional Theory (RDMFT) into a more competitive method \cite{TheoRDMFT,PvsE} (for more specific results see Refs.~\cite{TheoTrip,TheoGPC3,CBfunct}). In particular, it has been shown \cite{ExForce} for \emph{all} translationally invariant one-band lattice systems (regardless of their dimensionality, size and interactions) that the gradient of the exact universal functional diverges repulsively on the polytope boundary $\partial \mathcal{P}$. It is exactly this latter result and the suggested implications of (quasi)pinning which motivate us to explore and rigorously derive here the implications of pinning on the  respective $N$-fermion quantum state.

\subsection{Borland-Dennis setting: Implications of pinned occupation numbers}\label{subsec:BD}
We first discuss the implications of pinning within the specific Borland-Dennis setting, i.e.~for $(N,d)=(3,6)$. This in particular also allows us to understand how those implications may look like in the case of \emph{degenerate} NONs.

At first sight, expanding quantum states $\ket{\Psi}$ in the Borland-Dennis setting seems to require ${6 \choose 3}=20$ configurations $\bd{i}\equiv (i_1,i_2,i_3)$. By referring to the self-consistent expansion \eqref{PsiNO}, this reduces to just eight configurations, namely $(1,2,3),(1,2,4),(1,3,5),(1,4,5),$ $(2,3,6),(2,4,6),(3,5,6),(4,5,6)$. This result has been
communicated privately by Ruskai and Kingsley to Borland and Dennis (cf. Ref.~\cite{BD}) and represented an important ingredient for determining the respective GPCs (see also Ref.~\cite{Rus2}). In particular, the three equalities \eqref{d=6a} follow immediately.
In addition, the complex-valued coefficients $c_{\bd{i}}$ need to fulfil additional self-consistency conditions to ensure that the
corresponding one-particle reduced density operator $\rho_1$ \eqref{1RDO} is diagonal with respect to the natural orbitals.

In the following, we use $n_4,n_5,n_6$ as the independent variables in the occupation number pictures and the remaining ones follow form the conditions \eqref{d=6a}. Let us now assume that the NONs are saturating the GPC \eqref{d=6b},
\begin{equation}
0 = D(\bd{n})\equiv -n_4+n_5+n_6\,.
\end{equation}
This implies $c_{356}=c_{456}=0$ (see Theorem 3 in \cite{CSQuasipinning}) and the most general quantum state $\ket{\Psi}$ with pinned NONs therefore takes the form
\begin{eqnarray}\label{PsiBDpin}
\ket{\Psi}&=& c_{123}\ket{1,2,3}+c_{124}\ket{1,2,4}+c_{135}\ket{1,3,5}+c_{145}\ket{1,4,5} \nonumber \\
&&+c_{236}\ket{2,3,6}+c_{246}\ket{2,4,6}\,.
\end{eqnarray}
Using Eq.~\eqref{NONc}, the NONs follow as
\begin{eqnarray}\label{BDdiag}
n_4 &=& |c_{124}|^2+|c_{145}|^2+|c_{246}|^2 \nonumber \\
n_5 &=& |c_{135}|^2+|c_{145}|^2 \nonumber \\
n_6 &=& |c_{236}|^2+|c_{246}|^2 \,
\end{eqnarray}
and the requirement on the off-diagonal entries of $\rho_1$ to vanish read
\begin{eqnarray}\label{BDoff}
  0 &=& \bra{1}\rho_1\ket{6} \,\,\,\,\,= c_{123} c_{236}^\ast +c_{124} c_{246}^\ast  \label{off16}\\
  0 &=& -\bra{2}\rho_1\ket{5}= c_{123} c_{135}^\ast + c_{124} c_{145}^\ast \label{off25}\\
  0 &=& \bra{3}\rho_1\ket{4}\,\,\,\,\,= c_{123} c_{124}^\ast + c_{135} c_{145}^\ast + c_{236} c_{246}^\ast \label{off34}\,.
\end{eqnarray}
All the other off-diagonal entries vanish automatically and thus do not impose any conditions on the expansion coefficients $c_{\bd{i}}$.

\subsubsection{Non-degenerate NONs}
To illustrate the consequences of pinning, we use Theorem 4 from Ref.~\cite{CSQuasipinning} which states for all $\ket{\Psi}$ in the Borland-Dennis setting with $n_3 > n_4$
\begin{equation}\label{oldCS4}
|c_{124}|^2+|c_{135}|^2+|c_{236}|^2  \leq \frac{D(\bd{n})}{n_3-n_4} + 3D(\bd{n})\,.
\end{equation}
Thus, whenever $\ket{\Psi}$ exhibits pinning with $n_3 \neq n_4$, it takes the form
\begin{equation}\label{PsiBDnone}
\ket{\Psi} = c_{123}\ket{123}+c_{145}\ket{145}+c_{246}\ket{246}\,.
\end{equation}
Clearly, this includes the case of non-degenerate NONs, $\frac{1}{2}> n_4 > n_5 > n_6$.


\subsubsection{Degenerate NONs}
In general, understanding the implications of pinning for degenerate NONs turns out to be rather challenging. There are two reason for this: First, there is no unique natural orbital basis anymore and it is therefore not clear whether a selection rule of the form \eqref{PsiBDnone} may refer to \emph{all} possible natural orbital bases or to just one of them. Second, the saturation of some GPC $D \geq 0$ and an additional ordering constraint may automatically enforce the saturation of additional GPCs. In that case, the corresponding selection rule for the saturation $D\equiv 0$ might be more restrictive than in the case of non-degenerate NONs. The latter happens in the Borland-Dennis setting in case of a degeneracy $n_4=n_5$ (and assuming $n_3 \neq n_4$):
The GPC \eqref{d=6b} implies $n_6=0$ and thus \eqref{PsiBDpin} simplifies according to $c_{246}=0$ (recall Eq.~\eqref{NONc}).

The case of an $n_3=n_4$ degeneracy is conceptually different. First of all, result \eqref{oldCS4} does not apply anymore. Moreover, corresponding states $\ket{\Psi}$ could take the specific form \eqref{PsiBDnone} only with respect to highly distinctive bases of natural orbitals. Indeed, for any $\ket{\Psi}$ with $n_3=n_4$ of the form \eqref{PsiBDnone} there are infinitely many allowed orbital rotations in the $n_3=n_4$ subspace (leaving $\rho_1$ invariant) and changing the form $\eqref{PsiBDnone}$ to \eqref{PsiBDpin}, i.e.~leading to a superposition of six rather than three configurations. Yet, the converse turns out to be true as well. Given an arbitrary quantum state with pinned NONs and an $n_3=n_4$ degeneracy, expressed self-consistently according to \eqref{PsiBDpin} with respect to some choice of natural orbitals $\mathcal{B}_1=\{\ket{j}\}_{j=1}^6$. Then, there exists a (unitary) transformation of the natural orbitals
\begin{equation}
\ket{3}\rightarrow \ket{\tilde{3}}\,,\quad \ket{4}\rightarrow \ket{\tilde{4}}\,,\quad \ket{i}\rightarrow \ket{\tilde{i}}=\ket{i}\,,i=1,2,5,6\,,
\end{equation}
such that the state $\ket{\Psi}$ takes the form \eqref{PsiBDnone} with respect to the alternative choice $\mathcal{B}'_1=\{\ket{j'}\}_{j=1}^6$ of natural orbitals. The existence of such a unitary transformation of the degenerate natural orbitals
follows directly from the conditions \eqref{BDoff} and $n_3=n_4$. The reader may verify that this transformation
takes the form
\begin{equation}
\ket{\tilde{3}}\equiv \frac{c_{123}\ket{3}+c_{124}\ket{4}}{\sqrt{|c_{123}|^2+|c_{124}|^2}}\,,\quad \ket{\tilde{4}}\equiv \frac{c_{124}^\ast\ket{3}-c_{123}^\ast\ket{4}}{\sqrt{|c_{123}|^2+|c_{124}|^2}}\,,
\end{equation}
leading to
\begin{eqnarray}\label{PsiBD34}
\ket{\Psi} &=&\tilde{c}_{123}\ket{\tilde{1}\tilde{2}\tilde{3}}+\tilde{c}_{145}\ket{\tilde{1}\tilde{4}\tilde{5}}+
\tilde{c}_{246}\ket{\tilde{2}\tilde{4}\tilde{6}} \,. 
\end{eqnarray}
The corresponding transformed expansion coefficients follow as
\begin{eqnarray}
\tilde{c}_{123}&=&\sqrt{|c_{123}|^2+|c_{124}|^2}\,,\quad \tilde{c}_{145}=\frac{c_{124} c_{135}-c_{123} c_{145}}{\sqrt{|c_{123}|^2+|c_{124}|^2}}\nonumber \\
\tilde{c}_{246}&=&\frac{c_{124} c_{236}-c_{123} c_{246}}{\sqrt{|c_{123}|^2+|c_{124}|^2}}
\end{eqnarray}
and all the remaining coefficients $\tilde{c}_{\bd{i}}$ vanish.

\subsubsection{Geometric picture of the Borland-Dennis setting}
In the following we interpret these structure simplifications of the quantum state in case of pinning from a geometric point of view.
First, the polytope of attainable vectors $\bd{n}\equiv (n_4,n_5,n_6)$ is shown in Fig.~\ref{fig:fqmp} in gray.
\begin{figure}[h]
\centering
\includegraphics[width=0.4\textwidth]{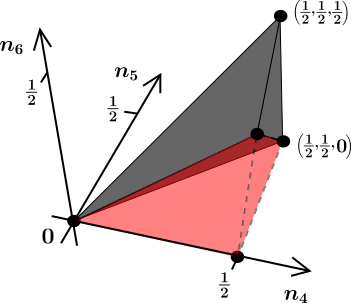}
\hspace{1.2cm}
\includegraphics[width=0.47\textwidth]{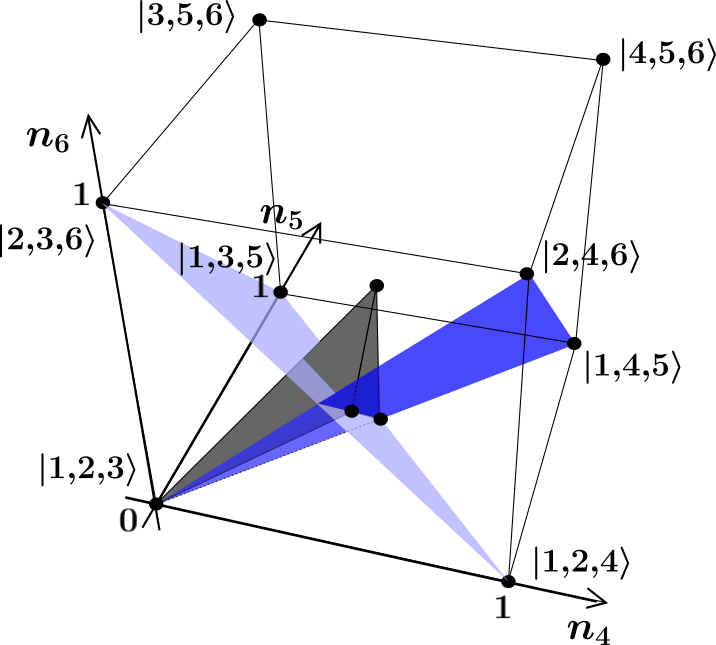}
\caption{Reduced polytope (gray) of possible (independent) NONs $(n_4,n_5,n_6)$ for the Borland-Dennis setting \ref{ex:BD}. Left: The Pauli simplex contains in addition the part shown in red, emphasizing that the GPCs are more restrictive than Pauli's original principle. Right: In case of pinning of the GPC \eqref{d=6b} only configurations $\bd{i}$ may contribute to $\ket{\Psi}$ that lie on the respective hyperplane (blue). In case of an additional degeneracy $n_3=n_4$ (i.e.~$n_4=\frac{1}{2}$) also those on the reflected hyperplane (light blue) may contribute (for more details see text).}
\label{fig:polytope}
\end{figure}
From the left side we can infer again that the GPCs are more restrictive than Pauli's exclusion principle constraints since the respective polytope $\mathcal{P}$ is a proper subset of the Pauli simplex $\Sigma$ (given by the polytope together with an extension shown in red). On the right side, the geometric picture as introduced in Section \ref{subsec:geom} is presented. The vector $\bd{n}$ of NONs follows as the center of mass of masses $|\bd{c}_{\bd{i}}|^2$ located at the positions $\bd{n}_{\bd{i}}$, the vertices of the Pauli hypercube. Their restrictions to the $(n_4,n_5,n_6)$-subspace are given by (recall Eq.~\eqref{NONconf})
\begingroup
\addtolength{\jot}{0.4em}
\begin{eqnarray}
&&\bd{n}_{123}=  \BDv{0}{0}{0}\,,\quad\bd{n}_{124}=  \BDv{1}{0}{0}\,,\quad\bd{n}_{135}=  \BDv{0}{1}{0}\,,\quad\bd{n}_{145}=  \BDv{1}{1}{0} \nonumber \\
&&\bd{n}_{236}=  \BDv{0}{0}{1}\,,\quad\bd{n}_{246}=  \BDv{1}{0}{1}\,,\quad\bd{n}_{356}=  \BDv{0}{1}{1}\,,\quad\bd{n}_{456}=  \BDv{1}{1}{1}\,.
\end{eqnarray}
\endgroup
In case of non-degenerate NONs, only the configurations $\bd{i}$ may contribute in the self-consistent expansion \eqref{PsiNO} according to \eqref{PsiBDnone} whose unordered spectra $\bd{n}_{\bd{i}}$ lie on the hyperplane corresponding to pinning (shown in blue). The same is still true in case of degeneracies $n_4=n_5$ or $n_5=n_6$. For a degeneracy $n_3=n_4$ (i.e.~$n_4=\frac{1}{2}$) and a generic choice of the natural orbitals in the $n_3=n_4$ subspace also the configurations $\bd{i}$ whose vectors $\bd{n}_{\bd{i}}$ lie on the light blue hyperplane may contribute. This latter hyperplane is given by the swapping $n_3 \leftrightarrow n_4$ of the blue hyperplane. Yet, according to \eqref{PsiBD34} there exists at least one basis $\mathcal{B}_1$ of natural orbitals with respect to which the weights on the light blue hyperplane are transformed away and would lie solely on the blue hyperplane.


The analysis of the Borland-Dennis setting suggests the following implications of pinning by a GPC $D>0$ in a general setting $(N,d)$: In case of non-degenerate NONs there is no ambiguity since the natural orbitals are unique and only those configurations $\bd{i}$ may contribute to $\ket{\Psi}$ whose unordered spectra $\bd{n}_{\bd{i}}$ (recall Eq.~\eqref{NONgeom}) lie on the respective hyperplane corresponding to pinning, $D\equiv 0$. In case of degenerate NONs there exist at least one basis $\mathcal{B}_1$ of natural orbitals with respect to which the original selection rule for non-degenerate NONs applies.

Although those main results of our work (see Theorems \ref{thm:pin1}, \ref{thm:pin2} and Corollaries \ref{cor:SR1}, \ref{cor:SR2} below) could be presented for both cases of non-degenerate and degenerate NONs together, we split them. This has the advantage that at least the results for non-degenerate NONs can be stated in a less technical form, namely not involving the ambiguity of natural orbital bases. For the proofs of various results we refer the reader to Part II.

\subsection{Implications of non-degenerate pinned occupation numbers}\label{subsec:pin}
In case of non-degenerate NONs the structural implications of pinning can be stated as
\begin{theorem}[Pinning of non-degenerate NONs]\label{thm:pin1}
Let $\ket{\Psi}\in \HN$ be an $N$-fermion quantum state whose non-degenerate NONs $\bd{n}$ saturate a GPC, $D(\bd{n})=0$ and denote the family of  $\ket{\Psi}$'s unique natural orbitals by $\mathcal{B}_1$. Then, $\ket{\Psi}$ lies in the zero-eigenspace of the respective $\hat{D}_{\mathcal{B}_1}$-operator (recall Definition \ref{def:Lhat}), i.e.
\begin{equation}
D(\bd{n})=0 \quad \Rightarrow \quad\hat{D}_{\mathcal{B}_1} \ket{\Psi} = 0\,.
\end{equation}
\end{theorem}
\noindent
It is worth noticing that Theorem \ref{thm:pin1} applies to various saturated GPC simultaneously.

Theorem \ref{thm:pin1} implies immediate structural simplifications for the state $\ket{\Psi}$ which are particularly well-pronounced in the self-consistent expansion \eqref{PsiNO} as already illustrated above:
\begin{corollary}[Selection rule for non-degenerate NONs]\label{cor:SR1}
Let $\ket{\Psi}\in \HN$ be an $N$-fermion quantum state whose non-degenerate NONs $\bd{n}$ saturate a GPC, $D(\bd{n})=0$. Then, only those configurations $\bd{i}$ may contribute in the self-consistent expansion \eqref{PsiNO} of $\ket{\Psi}$ whose unordered spectra $\bd{n}_{\bd{i}}$ (recall Eq.~\eqref{NONconf}) lie on the hyperplane defined by $D(\bd{n}_{\bd{i}})=0$. In other words, for each configuration $\bd{i}$ we have
\begin{equation}
D(\bd{n}_{\bd{i}})\neq 0\quad \Rightarrow \quad c_{\bd{i}} = 0\,.
\end{equation}
\end{corollary}

We present an example which illustrates Theorem \ref{thm:pin1} and the corresponding selection rule, Corollary \ref{cor:SR1}:

\begin{example}\label{ex:pin38}
We consider non-degenerate NONs in the setting $(N,d)=(3,8)$ that are saturating one of the GPCs, namely
\begin{equation}\label{GPC38}
D(\bd{n})\equiv9-19n_1-11n_2+21n_3+13n_4+5n_5+5n_6-3n_7-11n_8\geq0\,.
\end{equation}
According to Theorem \ref{thm:pin1} any corresponding quantum state $\ket{\Psi}$ has to lie in the zero-eigenspace of the respective $\hat{D}_{\mathcal{B}_1}$-operator,
\begin{equation}\label{Dhat38}
\underbrace{\big(9-19\hat{n}_1-11\hat{n}_2+21\hat{n}_3+13\hat{n}_4+5\hat{n}_5+5\hat{n}_6-3\hat{n}_7-11\hat{n}_8\big)}_{
\hat{D}_{\mathcal{B}_1}}\ket{\Psi} = 0\,.
\end{equation}
Corollary \ref{cor:SR1} then identifies all configurations which may contribute to the self-consistent expansion of $\ket{\Psi}$, namely
$\bd{i}=(1,2,3), (1,5,6), (1,3,8), (2,5,7), (5,7,8), (2,4,8), (1,4,7),$ $(2,6,7), (6,7,8)$. This reduction of ${8\choose 3}=56$
configurations to just $9$ highlights the remarkable implications of pinning as an effect in the one-particle picture on the structure of the corresponding many-fermion quantum state.
\end{example}

\subsection{Implications of degenerate pinned occupation numbers}\label{subsec:pindeg}
Based on the analysis of pinning by degenerate NONs in the Borland-Dennis setting (Section \ref{subsec:BD}) one may expect the following
generalization of Theorem \ref{thm:pin1} to degenerate NONs:
\begin{conjecture}\label{conj:pin2}
Let $\ket{\Psi}\in \HN$ be an $N$-fermion quantum state whose degenerate NONs $\bd{n}$ saturate some (possibly several) GPCs. Then, there exists an orthonormal basis $\mathcal{B}_1$ of natural orbitals such that $\ket{\Psi}$ lies in the zero-eigenspace of the respective $\hat{D}_{\mathcal{B}_1}$-operators of various saturated GPCs (recall Definition \ref{def:Lhat}), i.e.
\begin{equation}
\exists\mathcal{B}_1:\quad D(\bd{n})=0 \quad \Rightarrow \quad\hat{D}_{\mathcal{B}_1} \ket{\Psi} = 0\,.
\end{equation}
\end{conjecture}
\noindent There are actually a number of reasons (highlighted in Part II which presents various mathematical proofs) why the generalization of Theorem \ref{thm:pin1} to non-degenerate NONs and its proof are quite involved.

In the following we present a weaker extension of Theorem \ref{thm:pin1} to degenerate NONs. It refers to the saturation of exactly one GPC. Its proof requires in addition the validity of a technical assumption (presented as Assumption 13 in Part II)
which we could verify for all GPCs known so far. Hence, there is little doubt that the assumption is always valid and the corresponding addition to the following theorem might be unnecessary.

\begin{theorem}[Pinning of degenerate NONs]\label{thm:pin2}
Let $\ket{\Psi}\in \HN$ be an $N$-fermion quantum state whose degenerate NONs $\bd{n}$ saturate exactly one GPC, $D(\bd{n})=0$ and assume that the technical Assumption 13 from Part II is met. Then, there exists an orthonormal basis $\mathcal{B}_1$ of natural orbitals such that $\ket{\Psi}$ lies in the zero-eigenspace of the respective $\hat{D}_{\mathcal{B}_1}$-operator (recall Definition \ref{def:Lhat}), i.e.
\begin{equation}
D(\bd{n})=0 \quad \Rightarrow \quad \exists\mathcal{B}_1:\quad\hat{D}_{\mathcal{B}_1} \ket{\Psi} = 0\,.
\end{equation}
\end{theorem}
\noindent Despite the ambiguity of the natural orbital basis $\mathcal{B}_1$ it is worth recalling that the natural orbitals $\{\ket{j}\}$ are still referring to the non-increasingly ordered NONs (see also \eqref{1RDOdiag}).

In complete analogy to Theorem \ref{thm:pin1} and Corollary \ref{cor:SR1}, Theorem \ref{thm:pin2} implies immediately a corresponding selection rule identifying all configurations which may contribute to $\ket{\Psi}$ in case of pinning:
\begin{corollary}[Selection rule for degenerate NONs]\label{cor:SR2}
Let $\ket{\Psi}\in \HN$ whose degenerate NONs $\bd{n}$ saturate exactly one GPC, $D(\bd{n})=0$ and assume that the
technical Assumption 13 from Part II is met.  Then, there exists an orthonormal basis $\mathcal{B}_1$ of natural orbitals such that only configurations $\bd{i}$ may contribute to the self-consistent expansion \eqref{PsiNO} of $\ket{\Psi}$ whose unordered spectra $\bd{n}_{\bd{i}}$ (recall Eq.~\eqref{NONconf}) lie on the
the hyperplane $D\equiv 0$, i.e.
\begin{equation}
\exists \mathcal{B}_1\,\mbox{such that}\,\forall \bd{i}:\qquad D(\bd{n}_{\bd{i}})\neq 0 \quad \Rightarrow \quad c_{\bd{i}}=0\,.
\end{equation}
\end{corollary}

\subsection{Converse selection rule: Rationalizing pinning-based multiconfigurational ansatzes}\label{subsec:conv}
The remarkable implications of pinning as an effect in the one-fermion picture on the structure of the $N$-fermion quantum state
offers an alternative characterization of some existing variational post-Hartree-Fock ansatzes and suggests additional new ones: Each face $F$ of the polytope $\mathcal{P}$, as characterized by a certain number of saturated Pauli constraints, generalized Pauli constraints and ordering constraints $n_{i}-n_{i+1}\geq 0$, defines a state manifold $\mathcal{M}_F$ of quantum states. These are exactly those states $\ket{\Psi}\in \HN$ whose NONs map to the face $F$,
\begin{equation}\label{MF}
\mathcal{M}_F \equiv \big\{\ket{\Psi}\in \HN\,\big|\,\mbox{spec}^\downarrow(N \mbox{Tr}_{N-1}[\ket{\Psi}\!\bra{\Psi}]) \in F\big\}\,.
\end{equation}
Minimizing the energy expectation value of a given Hamiltonian $H$ of a system of interacting fermions over $\mathcal{M}_F$ then defines a variational scheme associated with the face $F$ with a corresponding variational energy
\begin{equation}\label{EMF}
E_{\mathcal{M}_F} \equiv \min_{\ket{\Psi} \in \mathcal{M}_F} \bra{\Psi}H \ket{\Psi}\,.
\end{equation}
From a qualitative point of view, one can say that the higher dimensional the face $F$, the higher dimensional the corresponding
state manifold $\mathcal{M}_F$ and thus the more computationally demanding the respective ansatz. Some well-known examples for such polytope face-associated variational schemes are the \emph{Complete Active Space Self-Consistent Field} (CASSCF) ansatzes (see, e.g., Refs.~\cite{CASSCF1,CASSCF2,CASSCF3,CASSCF4}). Indeed, according to Theorem \ref{thm:CAS} they can be characterized by the saturation of a certain number of Pauli exclusion principle constraints. Our main results, Theorems \ref{thm:pin1}, \ref{thm:pin2} and the respective selection rules, Corollaries \ref{cor:SR1}, \ref{cor:SR2}, highlight that even more elaborated variational ansatzes can be introduced by referring not only to the saturation of Pauli constraints but to extremal one-fermion information in general, i.e., pinned NONs.
The motivation for proposing such generalizations of CASSCF ansatzes is twofold. On the one hand, the study of smaller atoms \cite{CSQChem} has reveled that the GPCs have an additional significance for ground states beyond the one of the Pauli exclusion principle constraints, as quantified by the $Q$-parameter \cite{CSQ}. On the other hand, not all configurations $\bd{i}$ within a complete active space are relevant and it would be preferable to identify only the most significant ones. The gain in computational time could be used to increase the basis set size, allowing one to recover more of the dynamic correlation.

A comment is in order concerning the practical implementation of such variational schemes. After having fixed $F$, i.e.~the corresponding family of contributing configurations $\bd{i}$, one would minimize both the respective expansion coefficients $c_{\bd{i}}$ and the involved natural orbitals $\{\ket{j}\}_{j=1}^d$. Such variational approaches are known in quantum chemistry as \emph{Multiconfiguration Self-Consistent Field} (MCSCF)-ansatzes (see, e.g., the textbook \cite{pinkbook}). Yet, the stringent use of pinning-based variational ansatzes in the form \eqref{EMF} would be quite challenging and not particularly  efficient. This is due to the fact that the selection rule \ref{cor:SR1} defines $\mathcal{M}_F$ by referring to the self-consistent expansion \eqref{PsiNO}, i.e.~rather involved self-consistency conditions on the expansion coefficients $\bd{c}_{\bd{i}}$ would need to be imposed. From a converse point of view, an \emph{arbitrary} superposition of all allowed configurations $\bd{i}$ is typically not self-consistent. Hence, its relation to the face $F$ seems to be rather loose, since its vector $\bd{n}$ of non-increasingly ordered NONs lies actually in the interior of the polytope $\mathcal{P}$ rather than on the face $F$. To illustrate this, let us revisit Example \ref{ex:pin38}. We pick a random (real-valued) superposition of the allowed configurations listed in Example \ref{ex:pin38}:
$c_{123}=-0.2595, c_{156}=0.1877, c_{138}=-0.5043, c_{257}=$ $-0.1258, c_{578}=-0.0411, c_{248}=-0.6256, c_{147}=-0.0154, c_{267}=-0.1317, c_{678}=0.4660$. The corresponding vector $\bd{n}$ of decreasingly-ordered NONs follows as
\begin{eqnarray}\label{38NONs}
\bd{n} &=&  \{0.9418, 0.4140, 0.3914, 0.3569, 0.3215, 0.2696, 0.2521, 0.0527\} \nonumber \\
\downarrow & & \hspace{-0.65cm}\mbox{permutation}\,\pi \\
\bd{n}' &=& \{0.3569, 0.9418, 0.3215, 0.3914, 0.2696, 0.0527, 0.2521, 0.4140\} \nonumber\,.
\end{eqnarray}
For the GPC \eqref{GPC38} at hand, one finds $D(\bd{n})=1.0325>0$, i.e.~$\bd{n} \in \mathcal{P}$ lies far away from the polytope facet
defined by $D\equiv 0$. This is actually quite different for the vector $\bd{n}'$ obtained by permuting the NONs according to some specific permutation $\pi$. Of course, $\bd{n}'$ does not lie in the polytope $\mathcal{P}$ anymore since its entries are not properly ordered. Yet, by extending the face $D\equiv 0$ of $\mathcal{P}$ to a hyperplane in the space of all occupation number vectors (including the ones which are not decreasingly ordered), $\bd{n}'$ turns out to lie on that hyperplane, $D(\bd{n}')=D(\pi(\bd{n}))= 0.0000$. This is rather astonishing in particular since the one-particle reduced density matrices of such arbitrary superpositions are not diagonal in the original reference basis anymore. For instance, one finds for the superposition above $\bra{2}\rho_1 \ket{8}=-0.1870$. This surprising example has actually a deep origin:
\begin{theorem}[Converse selection rule]\label{thm:conv}
Let $F$ be a face of the polytope of the setting $(N,d)$ defined by the saturation of a specific family $\{D_k\}_{k \in K_F}$ of GPCs $D_k \geq 0$. For
an orthonormal basis $\mathcal{B}_1=\{\ket{j}\}_{j=1}^d$ of $\Ho$ we define
\begin{equation}
\mathcal{A}_F^{(\mathcal{B}_1)}:={\rm Span}\left\{\ket{\bd{i}}\,\big|\,\forall k\in K_F:\, D_k(\bd{n}_{\bd i})=0\right\}\,,
\end{equation}
i.e.~the vector space of all superpositions of configurations $\bd{i}$ fulfilling the selection rule \ref{cor:SR1} with respect to the basis $\mathcal{B}_1$ for all GPCs $D_k$ with $k \in K_F$.
Then, for any $\kpsi\in \mathcal{A}_F^{(\mathcal{B}_1)}$ there exists a basis $\mathcal{B}'_1$ of (possibly wrongly ordered) natural orbitals of $\ket{\Psi}$ such that all configurations $\bd{i}\in\mathrm{Supp}_{\mathcal{B}'_1}(\kpsi)$ also fulfil the selection rules $D_k(\bd{n}_{\bd{i}})=0$ for all $k \in K_F$. In particular, the corresponding vector $\bd{n}$ of (possibly wrongly ordered) NONs saturates $D_k\geq 0$ for all $k\in K_F$, i.e.~$\bd{n}$ lies on the hyperplane obtained by extending the face $F$ to non-decreasingly ordered occupation number vectors.
\end{theorem}
To illustrate the first part of this theorem we revisit Example \ref{ex:pin38}. Let $\mathcal{B}_1\equiv \{\ket{i}\}_{i=1}^8$ be some orthonormal basis for $\Ho$ and consider the GPC $D \geq 0$ from Example \ref{ex:pin38}. The corresponding linear space $\mathcal{A}_F^{(\mathcal{B}_1)} $ follows as (where $F$ denotes the face defined by $D\equiv0$)
\begin{eqnarray}\label{span38}
\mathcal{A}_F^{(\mathcal{B}_1)} &=& {\rm Span}\Big\{\ket{1,2,3},\ket{1,5,6},\ket{1,3,8},\ket{2,5,7},\ket{5,7,8},\ket{2,4,8}, \ket{1,4,7}, \nonumber \\
 && \qquad  \ket{2,6,7}, \ket{6,7,8}\Big\}\,.
\end{eqnarray}
Let $\ket{\Psi} \in \mathcal{A}_F^{(\mathcal{B}_1)}$, i.e.~$\ket{\Psi}$ is a linear combination of the nine specific configuration states $\ket{i_1,i_2,i_3}$ shown in \eqref{span38}. As already explained above, the corresponding one-particle reduced density matrix $\rho_1$ of $\ket{\Psi}$ is in general not diagonal with respect to $\mathcal{B}_1$, i.e.~its natural orbital basis is different than $\mathcal{B}_1$.
Naively one may thus expect that the self-consistent natural orbital expansion \eqref{PsiNO} of $\ket{\Psi}$ would involve all 56 configurations. Yet, the first part of Theorem \ref{thm:conv} states that this is not the case. In particular, there exists a permutation of $\ket{\Psi}$'s ordered natural orbitals yielding $\mathcal{B}_1'\equiv \{\ket{j'}\}_{j=1}^8$ with the effect that only those $\ket{\bd{j}'}\equiv \ket{j_1',j_2',j_3'}$
contribute to $\ket{\Psi}$ which fulfil the selection rule $D(\bd{n}_{\bd{j'}})=0$.

The converse selection rule \ref{thm:conv} establishes a more flexible relation between quantum states and polytope faces $F$ since it does not
refer to the self-consistent expansion \eqref{PsiNO} anymore. In particular, it therefore provides a solid foundation for more effective pinning-based MCSCF ansatzes minimizing the energy expectation value of a Hamiltonian $H$ over all states in
\begin{equation}
  \mathcal{A}_F \equiv \bigcup_{\tilde{\mathcal{B}}_1}\mathcal{A}_F^{(\tilde{\mathcal{B}}_1)} = \left\{u^{\otimes^N}\ket{\Psi}\,\Big| \,\ket{\Psi}\in \mathcal{A}_F^{(\mathcal{B}_1)},\,u:\Ho\rightarrow \Ho\,\mbox{unitary}\right\} \,.
\end{equation}
Such ansatzes are indeed MCSCF ansatzes in a strict sense: In a first step, one identifies (via the choice of a face $F$) a specific set of configurations $\bd{i}$ contributing to $\ket{\Psi}$. Then, in a second step one minimizes the energy expectation value with respect to various expansion coefficients (without any additional constraints on them) and all possible orbital choices $\mathcal{B}_1$.
The corresponding variational energy
\begin{equation}\label{EMF}
E_{\mathcal{A}_F} \equiv \min_{\ket{\Psi} \in \mathcal{A}_F} \bra{\Psi}H \ket{\Psi} \leq E_{\mathcal{M}_F}\,.
\end{equation}
is at least as good as the original one ($E_{\mathcal{M}_F}$) and the computational effort is significantly reduced by omitting the quadratic self-consistency conditions required in the characterization of $\mathcal{M}_D$.

It will be one of the future challenging to implement and test such pinning-based MCSCF ansatzes (for a proof of concept see \cite{CSHF}). In particular, one needs to develop a systematic procedure for identifying the appropriate polytope faces $F$, e.g., in the form of a renormalization group-inspired scheme which exploits the inclusion hierarchy of faces of different dimensionalities \cite{CSMCSCF}.

\subsection{Presence of pinning reveals symmetries of quantum states}
According to Theorem \ref{thm:pin1} and its generalization including the case of degenerate NONs, Theorem \ref{thm:pin2}, pinning $D(\bd{n})=0$ implies that the $N$-fermion quantum state $\ket{\Psi}$ lies in the zero-eigenspace of the corresponding natural orbital induced operator $\hat{D}_{\mathcal{B}_1}$. This means nothing else than that $\hat{D}_{\mathcal{B}_1}$ is the generator of a continuous symmetry of $\ket{\Psi}$,
\begin{equation}
e^{i \theta \hat{D}_{\mathcal{B}_1}} \ket{\Psi} = \ket{\Psi}\,.
\end{equation}
This symmetry could be a hidden symmetry of the state $\ket{\Psi}$ itself or a symmetry of the Hamiltonian.

We present a prominent example for a pinned quantum state. For the Hubbard model with three sites and three electrons the ground state was shown to exhibit pinning \cite{CS2015Hubbard}. In the self-consistent expansion \eqref{PsiNO} it takes the form \eqref{PsiBDnone}. The corresponding natural orbitals are given by the following spin-momentum states $\ket{k\sigma}$ ($k=0,1,2, \sigma=\uparrow\!/\!\downarrow$)
\begin{equation}
\ket{1}=\ket{0\!\uparrow},\,\ket{2}=\ket{1\!\uparrow},\,\ket{3}=\ket{0\!\downarrow},\,\,
\ket{4}=\ket{2\!\uparrow},\,\ket{5}=\ket{2\!\downarrow},\,\ket{6}=\ket{1\!\downarrow}\,.
\end{equation}
The corresponding natural orbital induced operator thus reads
\begin{equation}
\hat{D}_{\Psi} = 2\mathds{1}-\hat{n}_1-\hat{n}_2-\hat{n}_4=2\mathds{1}-\hat{n}_{0\uparrow}-\hat{n}_{1\uparrow}-\hat{n}_{2\uparrow}=
\frac{\mathds{1}}{2}-\frac{\hat{S}_z}{\hbar}\,.
\end{equation}
The presence of pinning in the Hubbard trimer reflects the system's $SU(2)$-symmetry, generated by the total spin $\hat{S}_z$ along the $z$-axis. It would be interesting to explore the meaning of those symmetry operators, e.g.~for the harmonic trap systems shown to exhibit (approximate) pinning \cite{CSpair}.

\section{Summary and conclusion}\label{sec:conl}
The concept of active spaces simplifies the description of interacting quantum many-body systems by restricting to a neighbourhood of active orbitals around the Fermi level. The respective $N$-fermion wavefunction ansatzes can be characterized by the saturation of a certain number of Pauli constraints $0 \leq n_i \leq 1$, identifying the occupied core orbitals ($n_i=1$) and the inactive virtual orbitals ($n_j=0$).
By referring to the generalized Pauli constraints, completing Pauli's original exclusion principle, we have provided a natural generalization of the concept of active spaces: We have explained and comprehensively illustrated that the saturation of \emph{any} one-body $N$-representability condition defines a distinctive space of active electron configurations contributing to the wave function ansatz (see Theorems \ref{thm:pin1},\ref{thm:pin2},\ref{thm:conv} and the selection rules \ref{cor:SR1},\ref{cor:SR2}). In contrast to the traditional complete active spaces defined through the saturation of Pauli's exclusion principle constraints, the use of such generalized active spaces does not necessarily mean to neglect dynamical correlations since more orbitals may contribute while the number of contributing configurations is still restricted. In particular, the choice of appropriate generalized active spaces would identify in an efficient and systematic way the significant electron configurations (rather than taking all of them into account as in complete active space self-consistent field (CASSCF)-ansatzes).
The present Part I therefore provides the theoretical foundation for possible wavefuntion based methods exploiting the fruitful mathematical structure underlying the generalized Pauli constraints. From a practical point of view, to achieve the full potential of our more systematic multiconfigurational approach, more effort needs to be spent on the mathematical side to calculate the generalized Pauli constraints for larger system sizes.

Moreover, according to Theorems \ref{thm:pin1}, \ref{thm:pin2}, pinning as an effect in the one-particle picture reveals the presence of symmetries. Those could be global symmetries of the underlying Hamiltonian (as, e.g., for Hubbard model clusters) or symmetries of just the quantum state at hand. Consequently, the successful search of possible (quasi)pinning in quantum systems could reveal and characterize possible ground state symmetries.

\section*{Acknowledgments}
We thank D.Gross for helpful discussions. We also acknowledge financial support from the UK Engineering and Physical Sciences Research Council (Grant EP/P007155/1) and the German Research Foundation (Grant SCHI 1476/1-1) (CS), the National Science Centre, Poland under the grant SONATA BIS: 2015/18/E/ST1/00200 (AS),
the Excellence Initiative of the German Federal and State Governments (Grants ZUK 43 \& 81) (AL).
\\

\bibliographystyle{unsrt}
\bibliography{pin}

\end{document}